\documentclass[10pt,journal]{IEEEtran}

\ifCLASSINFOpdf
\else
\fi
\usepackage{array,booktabs}
\usepackage{algpseudocode}
\usepackage{algorithm}
\usepackage{amsfonts}
\usepackage{amsmath}
\usepackage{amssymb}
\usepackage{amsthm}
\usepackage{graphicx}
\usepackage[noadjust]{cite}
\usepackage{mathrsfs}
\usepackage{stmaryrd}
\usepackage{siunitx}
\usepackage{multicol}
\usepackage{tikz}
\usepackage{setspace}
\usepackage{pstricks, pst-node, pst-plot, pst-circ}
\usepackage{moredefs}
\usetikzlibrary{shapes}

\makeatletter
\newtheoremstyle{newdefinition}
  {3pt}
  {3pt}
  {}
  {1em}
  {\itshape}
  {:}
  {.5em}
  {\thmname{#1}\thmnumber{\@ifnotempty{#1}{ }#2}%
   \thmnote{ {\the\thm@notefont(\itshape#3)}}}
\makeatother
\theoremstyle{newdefinition}

\makeatletter
\newtheoremstyle{newlemma}
  {3pt}
  {3pt}
  {}
  {1em}
  {\itshape}
  {:}
  {.5em}
  {\thmname{#1}\thmnumber{\@ifnotempty{#1}{ }#2}%
   \thmnote{ {\the\thm@notefont(\itshape#3)}}}
\makeatother
\theoremstyle{newlemma}
\newtheorem{lemma}{Lemma}

\makeatletter
\newtheoremstyle{newtheorem}
  {3pt}
  {3pt}
  {}
  {1em}
  {\itshape}
  {:}
  {.5em}
  {\thmname{#1}\thmnumber{\@ifnotempty{#1}{ }#2}%
   \thmnote{ {\the\thm@notefont(\itshape#3)}}}
\makeatother
\theoremstyle{newtheorem}
\newtheorem{theorem}{Theorem}

\makeatletter
\newtheoremstyle{newproposition}
  {3pt}
  {3pt}
  {}
  {1em}
  {\itshape}
  {:}
  {.5em}
  {\thmname{#1}\thmnumber{\@ifnotempty{#1}{ }#2}%
   \thmnote{ {\the\thm@notefont(\itshape#3)}}}
\makeatother
\theoremstyle{newproposition}

\makeatletter
\newtheoremstyle{newproof}
  {3pt}
  {3pt}
  {}
  {2em}
  {\itshape}
  {:}
  {.5em}
  {\thmname{#1}\thmnumber{\@ifnotempty{#1}{ }#2}%
   \thmnote{ {\the\thm@notefont(\itshape#3)}}}
\makeatother
\theoremstyle{newproof}

\providelength{\AxesLineWidth}       
\providelength{\plotwidth}           
\providelength{\LineWidth}           
\providelength{\MarkerSize}          
\newrgbcolor{GridColor}{0.8 0.8 0.8}%

\begin{document}
\title{Decentralized Hypothesis Testing in \\Energy Harvesting Wireless Sensor Networks}
\author{Alla~Tarighati,
      James~Gross,~\IEEEmembership{Senior Member,~IEEE,}
	 and~Joakim~Jald{\'e}n,~\IEEEmembership{Senior Member,~IEEE}
%
%
%
}
\maketitle

\begin{abstract}
We consider the problem of decentralized hypothesis testing in a network of energy harvesting sensors, where sensors make noisy observations of a phenomenon and send quantized information about the phenomenon towards a fusion center. The fusion center makes a decision about the present hypothesis using the aggregate received data during a time interval. We explicitly consider a scenario under which the messages are sent through parallel access channels towards the fusion center. To avoid limited lifetime issues, we assume each sensor is capable of harvesting all the energy it needs for the communication from the environment. Each sensor has an energy buffer (battery) to save its harvested energy for use in other time intervals. Our key contribution is to formulate the problem of decentralized detection in a sensor network with energy harvesting devices. Our analysis is based on a queuing-theoretic model for the battery and we propose a sensor decision design method by considering long term energy management at the sensors. We show how the performance of the system changes for different battery capacities. We then numerically show how our findings can be used in the design of sensor networks with energy harvesting sensors.

\end{abstract}

\IEEEpeerreviewmaketitle

\section{Introduction}\label{sec:intro}
Distributed detection problem formulations have traditionally addressed detection in sensor networks by considering network performance measures like error probability and receiver operating characteristic \cite{Varsh96}. In these setups, spatially separated sensors make observations of the same phenomenon and send a summary of their observations towards a fusion center (FC) through rate-constrained channels. Each sensor can be viewed as a quantizer which quantizes its observation, and according to the network arrangement, sends its output either to another sensor or to the FC. In many applications, sensors send their outputs to the FC through a multiple access channel \cite{li07,Ciu12,Ciu15} or through parallel access channels, commonly known as the parallel topology \cite{Lon90,Alla14}. A survey of early works on decentralized hypothesis testing in wireless sensor networks can be found in \cite{Veer12,Cham07,Chen06}.

A large number of sensors with small batteries and limited life-time are often used in wireless sensor networks. A major limitation of these sensors is their finite lifetime. In other words, the sensors work as long as their battery last and this implies that also the network has a limited lifetime. Many solutions to increase the lifetime of battery-powered sensor nodes have been proposed, see \cite{Aky02,Bae04,Nug06} and references therein. While in all of these methods the aim is to find an energy usage strategy to maximize the lifetime of a network, the lifetime remains bounded and finite. An alternative way of dealing with this problem is to use energy harvesting devices at the sensor nodes. An energy harvesting device is capable of acquiring energy from nature or from man-made sources \cite{Sha10,Ulu15}. 

Energy harvesting technologies provide a promising future for wireless sensor networks, such as self sustainability and an effectively perpetual network lifetime which is not limited by the sensor battery lifetime \cite{Ulu15,Gun14,Sud11}. While acquiring energy from the environment makes it possible to deploy wireless sensor networks in situations which are impossible using conventional battery-powered sensors, it poses new challenges related to the management of the harvested energy. These new challenges are due to the fact that the amount of energy available at a sensor is random, since the source of energy might not be available at all times we may want to use the sensor nodes.

We address the problem of detection in networks of sensors arranged in parallel, where each sensor is an energy harvesting device. At each time $t=1,2,\ldots$ the sensors send a message towards the FC about the state of the current hypothesis $H_t$ and the FC makes a decision about the hypothesis at that time. The sensors communicate with the FC using energy asymmetric on-off keying (OOK), as a low communication rate scheme for distributed detection applications \cite{Rago96}, where a positive message can be sent at the cost of one unit of energy, and a negative message is conveyed through a non-transmission at no cost in energy \cite{Berg09,Ciu13}. It was previously shown in \cite{Li11} that OOK is the most energy efficient modulation scheme under Rayleigh fading non-coherent transmission, though we are not primarily concerned with explicitly modeling the fading channels between the sensors and the FC herein. We assume that each sensor is equipped with an internal battery and is allowed to use a long term energy conversion policy. We assume the observations at the sensors are, conditioned on the true hypothesis, independent and our goal is to design the sensors' transmission rules in such a way that the error probability at each time instance $t$, at the FC, is minimized. To this end, we use the Bhattacharyya distance between the conditional distributions at the FC input, which has been frequently used in the past as a performance measure in the design of distributed detection systems \cite{Poor77,Lon90,Alla15O}. 

Using the Bhattacharyya distance, essentially reduces the joint design of decision rules at the sensors to the design of decision rules at single sensors. Doing so may risk not capturing all the aspects of the joint design of the decision rules at peripheral sensors. However, due to the analytical tractability that follows from this choice of design rule, the Bhattacharyya distance has been frequently considered as a performance metric before in the literature \cite{Lon90}.

The novelty of our work is in the formulation of a decentralized detection problem with system costs coupled to the random behavior of the energy available at the sensors. Concretely, we will find the depletion probability at the sensor batteries, and evaluate the performance of the network for different battery capacities (buffer sizes). We will show how the problem formulation changes (compared to the unconstrained case) when we consider the energy features in the problem of designing the sensors in the network.

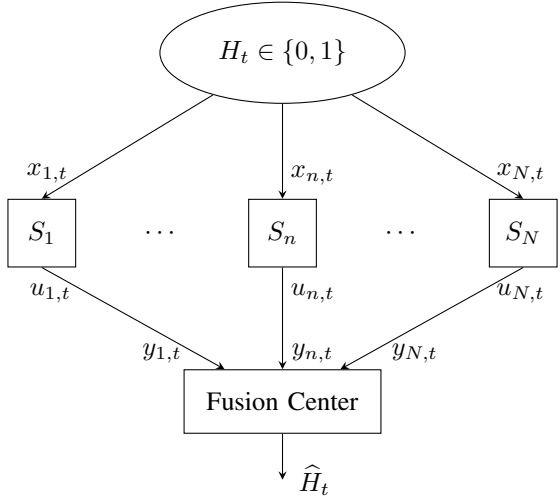
\begin{figure}
\centering
\begin{tikzpicture}[align=center,scale=0.8,>=stealth] 
\node (FC) at (0,0.2) [rectangle,inner sep=3mm,draw] {Fusion Center};
\node (DM1) at (-4,3) [rectangle,minimum size=0.9cm,draw] {$S_1$};
\node (DM2) at (-2,3)  {$\cdots$};
\node (DMN1) at (2,3) {$\cdots$};
\node (DMN) at (4,3) [rectangle,minimum size=0.9cm,draw] {$S_N$};
\node (DMn) at (0,3)  [rectangle,minimum size=0.9cm,draw] {$S_{n}$};
\node (PH) at (0,6) [ellipse,inner sep=3mm,draw] {$H_t\in\{0,1\}$};
\draw [->] (DM1.south) -- (FC) node [near start,left,inner sep=6pt] {$u_{1,t}$};
\draw [->] (DMn.south) -- (FC) node [near start,right,inner sep=3pt] {$u_{n,t}$};
\draw [->] (DMN.south) -- (FC) node [near start,right,inner sep=7pt] {$u_{N,t}$};
\draw [->] (FC.south) -- (0,-1.1) node [right,inner sep=6pt] {$\widehat{H}_t$};
\draw [->] (PH) --  (DM1.north) node [near end,left,inner sep=6pt] {$x_{1,t}$};
\draw [->] (PH) -- (DMn.north) node [near end,right,inner sep=3pt] {$x_{n,t}$};
\draw [->] (PH) -- (DMN.north) node [near end,right,inner sep=6pt] {$x_{N,t}$};

\node at (-2,1) {$y_{1,t}$};
\node at (.5,1) {$y_{n,t}$};
\node at (2.2,1) {$y_{N,t}$};
\end{tikzpicture}
\caption{Decentralized hypothesis testing scheme in a parallel network.}
\label{fig:topology}
\end{figure}

Distributed inference using energy harvesting agents has gained a lot of interest during the last years (see \cite{Nay13,Zhao13,Huang13,Nour15,Li16,Hong14,Liu14,Cha12,Med10,Alla16-EUSIPCO} and references therein). In the context of distributed estimation, Nayyar \emph{et al.} studied the structure of optimal communication scheduling \cite{Nay13}, Zhao \emph{et al.} \cite{Zhao13}, Huang \emph{et al.} \cite{Huang13}, and Nourian \emph{et al.} \cite{Nour15} studied optimal power allocation schemes, while Li \emph{et al.} \cite{Li16} studied the performance of an energy harvesting relay-aided cooperative network under fading channels. Hong \cite{Hong14} considered the case where each peripheral node sends a re-scaled version of its observation towards a FC, if the required energy for that transmission is available at the sensor; otherwise the sensor remains silent. This non-transmission also conveys information regarding the magnitude of the observed signal. Hua Liu \emph{et al.} \cite{Liu14} considered adaptive quantization is distributed estimation using a game-theoretic approach. Energy harvesting in relay systems was also considered in \cite{Cha12,Med10}, where the nodes (or relays) harvest the energy they need to transmit their amplified received messages towards a receiver. Medepally and Mehta \cite{Med10} considered a relay selection scheme where if multiple active relays are available, one of them is selected to transmit. In the context of distributed detection, however, finding optimal decision rules at the remote energy harvesting sensors is largely open. In \cite{Alla16-EUSIPCO} we studied the structure of decision rules at agents with infinite battery size and in error-free communication channels, while in this paper we shall extend the results to a more general case of erroneous communication channels.

The outline of this paper is as follows. In Section \ref{sec:system} we describe the structure of a parallel network and an energy harvesting sensor, and formulate the problem. In Section \ref{sec:design} we study the performance of an energy harvesting sensor with different battery capacities. In Section \ref{sec:results} we will illustrate our results in the design of sensor networks by presenting numerical simulations, and finally in Section \ref{sec:conclude} we conclude the paper.

\section{Preliminaries}\label{sec:system}
In this section we first present the system model. We then define an energy harvesting sensor and formulate the problem of designing the energy harvesting sensors in the network. Limitations of the proposed system model, and their influence on the obtained results, are further discussed in the concluding remarks in Section~\ref{sec:conclude}.
\subsection{System Model}
We consider a decentralized hypothesis testing problem where $N$ sensors $S_1,\ldots,S_N$ are arranged in parallel, according to Fig.~\ref{fig:topology}. During each time interval $t$ (defined as $[t,t+1)$) each sensor $S_n$, $n\in\{1,\ldots,N\}$, makes an observation $x_{n,t}$ from the same phenomenon and sends a message $u_{n,t}$ towards the FC. Note that $x_{n,t}$ is just assumed to be an element of an abstract space, so it could be a single scalar measurement or a vector of measurements made during time interval $t$ at sensor $S_n$. We consider the case where different observations at the sensor $S_n$, conditioned on the hypothesis, are independent and identically distributed and $X_n$ is a random variable corresponding to observations at sensor $S_n$.

At each time interval the phenomenon $H_t$ is modeled as a random variable drawn from a binary set $\{0,1\}$ with a-priori probabilities $\pi_0$ and $\pi_1$, respectively, and gives rise to conditionally independent observations $X_{n}\in\mathcal{X}_n$ with conditional distribution $f_{X_n\vert H_t}\left(x_{n,t}\vert h_t \right)$ for $h_t\in \{0,1\}$ at the sensors. In this paper, we assume that the present hypothesis $H_t$ changes over time in an i.i.d.\ fashion while it is fixed during each time interval. We also assume that the sensors are allowed to use a long term energy usage policy managed together with an internal battery. In other words, we consider the long term behavior of a sensor when its battery state operates in steady state.

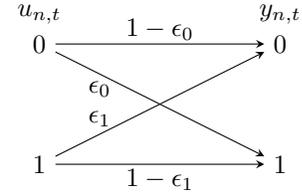
\begin{figure}
\centering
\begin{tikzpicture}[align=center,scale=0.8,>=stealth] 
\node (u0) at (0,2) {$0$};
\node (u1) at (0,0) {$1$};
\node (y0) at (4,2) {$0$};
\node (y1) at (4,0) {$1$};

\draw [->] (u0) -- (y0) node at (2,2.25) {$1-\epsilon_0$};
\draw [->] (u0) -- (y1) node at (1,1.25) {$\epsilon_0$};
\draw [->] (u1) -- (y0) node at (1,.75) {$\epsilon_1$};
\draw [->] (u1) -- (y1) node at (2,-0.25) {$1-\epsilon_1$};

\node at (0,2.5) {$u_{n,t}$};
\node at (4,2.5) {$y_{n,t}$};
\end{tikzpicture}
\caption{Binary asymmetric channels between sensors and the FC.}
\label{fig:BAC}
\end{figure}

The communication channels between the sensors and the FC are one-way links from the sensors, and there is no communication between the sensor nodes. The sensors communicate with the FC using an energy asymmetric on-off keying, where a positive message (labeled by ``$1$'') is conveyed by transmission of a message and a negative message (labeled by ``$0$'') is conveyed by a non-transmission. We model the channel between sensor $S_n$ and the FC by the most general binary discrete memoryless channel that is the so called \emph{binary asymmetric channel} (BAC) \cite{ber61}, where the $0$-to-$1$ error occurs with probability $\epsilon_0<0.5$ and the $1$-to-$0$ error occurs with probability $\epsilon_1<0.5$, as shown in Fig.~\ref{fig:BAC}, where $y_{n,t}\in\{0,1\}$ is the received message at the FC corresponding to the sensor output $u_{n,t}$. When $\epsilon_0=\epsilon_1$ the channel will be a binary symmetric channel (BSC) and when $\epsilon_0=0$, the channel will reduce to the $Z$-channel. It is important to stress that although we consider the sensor-to-FC channel to be unreliable, we also consider it to be fixed and not subject to alteration by the sensors, i.e., the optimization of the sensor decision rules does not encompass the design of the over the air transmission protocol.

The BAC is also a relevant \emph{high-level} model for the case where each sensor-to-FC channel is a fading channel and the FC (as a non-coherent receiver) uses an energy detector to detect its input $y_{n,t}$, i.e., where
\[y_{n,t}= \left\{ \begin{array}{ll}
          1 & \mathcal{E}_{n,t}> \Gamma\\
		0 & \text{Otherwise}\,, \end{array}\right.\]
where $\mathcal{E}_{n,t}$ is the energy received at the FC resulting from sensor $u_{n,t}$ and $\Gamma$ is some positive threshold \cite{Berg09}. For OOK transmission over a fading channel with an energy threshold detector at the FC, it may be reasonable to assume that $\epsilon_1 \gg \epsilon_0$, reflecting the belief that deep fades causing the received energy to not meet the specified threshold are more likely than noise induced fluctuations that push the decision metric over the threshold boundary.

The overall problem considered is structurally similar to the classical binary hypothesis testing problem over one bit channels. In other words, each sensor $S_n$, given a realization $x_{n,t}$ of $X_{n}$, computes a message $u_{n,t}\in \{0,1\}$ using its decision function $\gamma_n:\mathcal{X}\to \{0,1\}$, i.e.,
\begin{equation*}
\gamma_n(x_{n,t})=u_{n,t}\,,
\end{equation*}
and sends this message towards the FC. The FC, based on the aggregate received BAC channel outputs $\underline{y}_t\triangleq(y_{1,t},\ldots,y_{N,t})$, makes a decision $\widehat{h}_t\in\{0,1\}$ (with corresponding random variable $\widehat{H}_t$) about the present hypothesis at each time interval, using its decision function $\gamma_0:\{0,1\}^N\to \{0,1\}$, i.e.,
\begin{equation*}
\gamma_0(\underline{y}_t)=\widehat{h}_t\,.
\end{equation*}
Note that the FC does not directly have access to the sensor outputs $\underline{u}_t\triangleq(u_{1,t},\ldots,u_{N,t})$, while it has access to the corresponding BAC channel outputs $\underline{y}_t$.

The overall objective in this paper is to design the decision functions $\gamma_n$, for $n=0,\ldots,N$, of the FC and the sensors in such a way that some performance measure is optimized. This problem was considered extensively in the classical distributed detection literature \cite{Varsh96,Lon90,Alla14} when energy is always available at the sensors to send their messages and the sensor-to-FC channels are error-free channels. However, the problem of designing sensors' decision functions when each sensor is an energy harvesting device is largely open. We shall herein consider the problem of designing decision functions of the sensors when each sensor is an energy harvesting device. In what follows, we first define our performance metric and then, in Section~\ref{sec:EHS}, we define an energy harvesting device. 

Let $U_n$ be a random variable corresponding to output messages of sensor $S_n$. Then, due to the independence of observations, the conditional PMF associated with the message vector $\underline{U}\triangleq (U_1,\ldots,U_N)$ can be obtained according to
\begin{equation}\begin{split}
P_{\underline{U}\vert H_t}\left(\underline{u}\vert h_t \right)=\prod_{n=1}^N \int_{x\in \gamma^{-1}_n(u_n)}f_{X_n\vert H_t}\left(x\vert h_t \right)dx\,,
\end{split}\end{equation}
where $\underline{u}\triangleq\left(u_1,\ldots,u_N \right)$, and where $\gamma^{-1}_n(u_n)$ is the set of observations $x\in\mathcal{X}$ that satisfy $\gamma_n(x)=u_n$.

Let $\mathcal{B}_{\mathrm{Tot},t}$ be the total Bhattacharyya distance (BD) at the FC at time instance $t$, and for given sensor decision functions ${\gamma_1,\gamma_2,\ldots,\gamma_N}\,$:
\begin{equation*}
\mathcal{B}_{\mathrm{Tot},t}=-\log\sum_{\underline{y}_t\in\{0,1\}^N}\sqrt{P_{\underline{Y}\vert H_t}(\underline{y}_t\vert 0)P_{\underline{Y}\vert H_t}(\underline{y}_t\vert 1)}\,,
\end{equation*}
where $\underline{Y}\triangleq (Y_1,\ldots,Y_N)$ and $Y_n$ is a random variable corresponding to BAC channel output $y_{n,t}$. The Bayesian error probability $P_{\mathrm{E},t}\triangleq \Pr\left(\widehat{H}_t\neq H_t\right)$ is minimized when the FC applies the maximum a-posteriori (MAP) rule \cite{Tsi88} and the corresponding error probability can be upper bounded by the total BD according to \cite{Kai67}
\begin{equation*}
P_{\mathrm{E},t}\leq \sqrt{\pi_0\pi_1}\, e^{-\mathcal{B}_{\mathrm{Tot},t}}\,.
\end{equation*}
In the context of distributed hypothesis testing, it has been acknowledged that the Bayesian error probability criteria in most cases does not lead to a tractable design procedure for decision functions, even for small sized networks. This led authors to consider the Bhattacharyya distance as a performance measure of the network \cite{Lon90,Poor77,Alla15O}. In addition, when the observations at the sensors, conditioned on the true hypothesis, are independent, the total Bhattacharyya distance at the FC simplifies greatly and decouples \cite{Alla15O}, i.e., 
\begin{equation*}\begin{split}
\mathcal{B}_{\mathrm{Tot},t}&=\sum_{n=1}^N -\log\sum_{{y}_{n,t}=0,1}\sqrt{P_{{Y_n}\vert H_t}({y}_{n,t}\vert 0)P_{{Y_n}\vert H_t}({y}_{n,t}\vert 1)}\\
&=\sum_{n=1}^N\mathcal{B}_{n,t}\,.
\end{split}\end{equation*}
In other words, at each time instance $t$, the total Bhattacharyya distance at the FC of a parallel network is the summation of all the Bhattacharyya distances delivered from different sensors. Therefore, a network which maximizes the Bhattacharyya distance at the FC is a network with individually optimized sensors. Then, it suffices to maximize the BD received at the FC from each sensor, separately. From now on we will focus on a single energy harvesting sensor $S$ and drop the subscript $n$ and denote its BD at time $t$ by $\mathcal{B}_t$. Throughout this paper, we interchangeably use the terms ``Bhattacharyya distance" and ``BD".

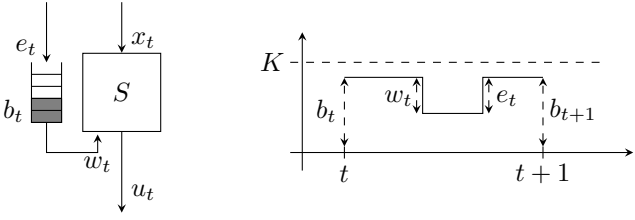
\begin{figure}[t]
\begin{center}
\begin{tikzpicture}[align=center,scale=0.8,>=stealth]
\def\sensor#1#2{
\begin{scope}[shift={#1}, rotate=#2]
\node (Sn) at (0,0) [rectangle,inner sep=4mm,draw] {$S$};
\draw [->] (0,1.5) -- (Sn) node [near end, right]{$x_{t}$};
\draw [->] (Sn) -- (0,-2) node [near end, right]{$u_{t}$};
\draw (-1,0.5) -- (-1,-.5) -- (-1.5,-0.5) -- (-1.5,0.5); 
\draw [fill=gray] (-1,-0.1) -- (-1,-.5) -- (-1.5,-0.5) -- (-1.5,-0.1) -- (-1,-0.1); \node at (-1.8,-0.3) {$b_t$};
\draw (-1,-.5) -- (-1.5,-0.5) ;
\draw (-1,-.3) -- (-1.5,-0.3) ;
\draw (-1,-0.1) -- (-1.5,-0.1) ;
\draw (-1,.1) -- (-1.5,0.1) ;
\draw (-1,.3) -- (-1.5,0.3) ;
\draw [->] (-1.25,1.5) -- (-1.25,0.5) node [near end, left]{$e_t$};
\draw [->] (-1.25,-0.5) -- (-1.25,-1) -- (-0.4,-1) -- (-0.4,-0.7) node at (-0.4,-1.2) {$w_t$};
\end{scope}
}

\def\interval#1#2{
\begin{scope}[shift={#1}, rotate=#2]
\draw [->] (2,-.8) -- (2,1.5); 
\draw [->] (1.8,-0.5) -- (7.5,-0.5);
\draw (2.7,.75) -- (4,.75) -- (4,0.15) -- (5,0.15) -- (5,0.75) -- (6,0.75);
\draw (2.7,-0.6) -- (2.7,-0.4) node at (2.7,-0.85) {$t$};
\draw (6,-0.6) -- (6,-0.4) node at (6,-0.85) {$t+1$};
\draw [<->,dashed] (2.7,-.4) -- (2.7,0.74) node at (2.4,0.2) {$b_t$};
\draw [<->,dashed] (6,-.4) -- (6,0.74) node at (6.5,0.2) {$b_{t+1}$};
\draw [<->,dashed] (3.9,0.15) -- (3.9,0.74) node at (3.6,0.4) {$w_{t}$};
\draw [<->,dashed] (5.1,0.15) -- (5.1,0.74) node at (5.4,0.4) {$e_{t}$};
\draw [dashed] (1.8,1) -- (7,1) node at (1.5,1) {$K$};
\end{scope}
}

\sensor{(0,0)}{0}
\interval{(1,-0.5)}{0}

\end{tikzpicture}
\end{center}
\caption{Model of an energy harvesting node (left) and an example of a time interval $t$ (right).}
\label{fig:EHS}\end{figure}
\subsection{Energy Harvesting Sensors}\label{sec:EHS}
Consider an energy harvesting sensor $S$ as in Fig.~\ref{fig:EHS}. Let us assume that during time interval $t$ energy $e_t$ arrives stochastically at the node as a stationary and ergodic random process $E_t$, see \cite{Ulu15}. Let $b_t$ be the battery state and let $B_t$ be its corresponding random process, which is in general a correlated random process over time even when energy harvesting process $E_t$ is i.i.d. Note that the actions of the sensor affect the future of the battery state, and the sensor knows the battery state. Assume that the battery size (capacity) is $K$. Then the amount of available energy in the battery at transmission time $t+1$ is (cf. \cite{Sha10})
\begin{equation}b_{t+1}=\min\left\{b_{t}-w_t+e_t,\,K\right\},\label{eq:bt1}\end{equation}
where $w_t$ is the amount of energy used to send the message $u_t$ at time $t$, with the corresponding random variable $W_t$. We assume that the arrival energy $e_t$ during time interval $t$ can not be used at the same time interval (see Fig.~\ref{fig:EHS}). The same as in \cite{Valen16a,Miche15,Miche12a,Tutu14}, we assume energy arrives in packets and at each time interval the sensor is capable of harvesting at most one packet of energy. We also assume $e_t\in\{0,1\}$ is drawn from a Bernoulli distribution, with $\Pr(e_t=1)=p_{e}$. We further assume that only sending a message costs a packet of energy, and the energy of making the observation and processing is negligible. These assumptions, while rather simplistic, are repeatedly used in the literature (see \cite{Valen16a,Miche15,Miche12a,Tutu14} and references therein) for an energy harvesting system. We adopt this model for the sake of analytical tractability and insight. Thus, $w_t=1$ if $u_t=1$ was sent during time interval $t$, otherwise $w_t=0$, i.e.,
\begin{equation*}
w_{t} =  w_{t}(u_t) = \begin{cases}
           1 & u_{t}=1\,,\\
		0 & u_{t}=0 \,. \end{cases}
\end{equation*} 
As noted before, the received BD at the FC from sensor $S$ at time $t$ is
\begin{equation}
\mathcal{B}_t=-\log\left[ \sum_{y_t=0,1}\sqrt{P_{Y\vert H_t}({y_t}\vert 0)P_{Y\vert H_t}({y_t}\vert 1)} \right],
\label{eq:BEH}\end{equation}
where
\begin{equation}\begin{split}
P_{Y\vert H_t}({y_t=i}\vert h)=&P_{U\vert H_t}({u_t=i}\vert h)(1-\epsilon_i)+\\
&P_{U\vert H_t}({u_t=1-i}\vert h)\epsilon_{1-i}\,.
\label{eq:PY}\end{split}\end{equation}

We say a sensor decision function $\gamma$ is a likelihood-ratio quantizer (or likelihood-ratio threshold) if
\begin{equation}
u_t=\gamma(x_t)\triangleq  \left\{ \begin{array}{ll}
          1 & l(x_t)\geq \Theta,\\
		0 & \text{Otherwise}\,, \end{array}\right.
\label{eq:trtest} 
\end{equation} 
for a given $\Theta$, where \[l(x_t)\triangleq \ln\frac{f_{X\vert H_t}(x_t\vert 1)}{f_{X\vert H_t}(x_t\vert 0)}\] 
is the log-likelihood ratio for a given observation $x_t$. Let $l_t\triangleq l(x_t)$ and let $L\triangleq l(X)$. 
It was shown in \cite{Tsi93Ext} that, for the \emph{unconstrained} case\footnote{By \emph{unconstrained} we mean the situation where energy is always available at the sensor.}, and error-free sensor-to-FC channels, likelihood-ratio quantizers are optimal decision rules at the local sensors. Chen and Willett \cite{Biao05} then generalized this result for the case of non-ideal sensor-to-FC channels. They have shown that likelihood-ratio thresholds are optimal for non-ideal channels as well (see \cite{Biao05,Bin06} for the conditions under which this result holds). Furthermore, for scalar observations $x_t$, when the likelihood-ratio of the observation at a sensor $S$ is monotone and increasing in $x_t$, the likelihood-ratio threshold can be directly translated to the sensor observation. This is an immediate consequence of the Karlin-Rubin theorem \cite{Leh06}.

Throughout this paper, we denote the optimal threshold of an energy unconstrained sensor by ${\Theta}^\star_{\mathrm{u}}$. In other words, ${\Theta}^\star_{\mathrm{u}}$ is the threshold in \eqref{eq:trtest} which maximizes \eqref{eq:BEH}, i.e.,
\begin{equation}\label{eq:BConven}\begin{split}
{\Theta}^\star_{\mathrm{u}}=\arg\max_{\Theta}\bigg\{-\log\Big[&\sqrt{[\epsilon_0+\delta q_0][\epsilon_0+\delta q_1]}\,+\\
&\sqrt{[1-\epsilon_0-\delta q_0][1-\epsilon_0-\delta q_1]} \,\Big]\bigg\}
\end{split}\end{equation}
where $$q_{h}\triangleq \Pr\left( L\geq \Theta\vert H_t=h \right),$$
$$\delta\triangleq 1-\epsilon_0-\epsilon_1,$$for $h=0,1$. Note that $q_0$ and $q_1$ are the false alarm and the detection probability of an unconstrained sensor, respectively.  

We say that the observation model at the sensor is \textit{separable}, or perfect, if there exists a threshold $\Theta$ for which 
\begin{equation} 
q_0=0\,,\quad q_1=1\,.
\label{eq:sObs}\end{equation}
In this situation, for an unconstrained sensor and error-free channel, the BD will be infinity and the detection problem is trivial. However, in real-world applications such observation models do not exist and the observation models are \textit{non-separable}, i.e., there is no threshold $\Theta$ for which the conditions in \eqref{eq:sObs} are both satisfied. However, such a condition may hold asymptotically in the high SNR limit. Thus, let us define $\mathcal{E}$ as the signal-to-noise ratio (SNR) at the sensor. We say that the observation model at the sensor is \textit{asymptotically separable} if, when SNR goes to infinity, for a sensor with non-separable observations there exists a threshold ${\Theta}$ for which
\begin{equation} 
\lim_{\mathcal{E}\to \infty}q_0=0\,,\quad \lim_{\mathcal{E}\to \infty}q_1=1\,,
\label{eq:px01}\end{equation}
and therefore the corresponding BD for the unconstrained sensor and error-free channel goes to infinity, $\lim_{\mathcal{E}\to \infty}\mathcal{B}_t=\infty$. Many observation models that are considered in the literature are asymptotically separable and in the following we will introduce one of these observation models. Consider the case where each observation is from a Rayleigh distribution with scale parameter $\sigma_0$ or from a Rician distribution with scale parameter $\sigma_1$ and noncentrality parameter $s$.  
The conditional distributions at the sensor are therefore
\begin{equation}\begin{split}
f_{X\vert H_t}\left( x\vert 0 \right)&=\frac{x}{\sigma^2_0}\exp\left({-\frac{x^2}{2\sigma^2_0}}\right), \text{ and}\\
f_{X\vert H_t}\left( x\vert 1\right)&=\frac{x}{\sigma^2_1}\exp\left({-\frac{x^2+s^2}{2\sigma^2_1}}\right)I_0\left(\frac{xs}{\sigma^2_1}\right),
\label{eq:dist01}\end{split}\end{equation}
where $I_0(z)$ is the modified Bessel function of the first kind with order zero. This is a relevant observation model for low complexity sensors in a wireless sensor network used to detect the presence of a known signal in Gaussian noise based on the received power. For simplicity, we assume $\sigma_0=\sigma_1=1$ and by definition $\mathcal{E}\triangleq  s^2$. Now, letting SNR go to infinity, there is a threshold $\Theta$ under which the conditions in \eqref{eq:px01} are satisfied, and therefore the observation model is asymptotically separable.

We shall herein assume that the energy constrained sensor $S$ is also a single-threshold likelihood-ratio quantizer that applies the following threshold test:
\begin{equation}
u_t=\gamma(x_t,b_t)\triangleq  \left\{ \begin{array}{ll}
          1 & l_t\geq \Theta,\, b_t>0,\\
		0 & \text{Otherwise}. \end{array}\right.
\label{eq:trtest2} 
\end{equation}
In other words, a sensor $S$ at each time $t$ compares the likelihood-ratio $l_t$ of its observation $x_t$ with a threshold $\Theta$. If $l_t\geq \Theta$ and the battery is not empty $b_t>0$, it sends a message towards the FC, otherwise it remains silent.\footnote{Note that the condition $b_t>0$ is a consequence of $w_t=1$ assumption.}  We would like to note that, the FC in this setup is a static device, in the sense that, it makes decision at each time $t$ only based on its input messages at the same time, and it does not use previous input messages from the sensors. In other words, the FC does not take into account any correlation in the received messages from a sensor that could (possibly) be introduced by the memory of the battery at that sensor. In what follows we will find a threshold $\Theta^\star$ which maximizes the delivered Bhattacharyya distance from a sensor to the FC. We will show that for the problem at hand an optimum decision rule not only depends on the observation model, it also depends on the battery charge state, arrival energy features and the capacity of the battery ${K}$. We further show that the resulting thresholds in general differ from those of the unconstrained case.

In the following section, for different battery capacities, we study the depletion probability of an energy harvesting sensor. Furthermore, we will show how one can design an energy harvesting sensor for different battery capacities. Note that by \emph{designing energy harvesting sensors} we mean the selection of the decision threshold $\Theta$ in \eqref{eq:trtest2}.

\section{Design of Energy Harvesting Sensors}\label{sec:design}
To formulate the problem, let us first find conditional mass probabilities resulting from a sensor decision $u_t\in\{0,1\}$.
\begin{equation}\begin{split}
P_{U\vert H_t}(u_t=1\vert h)&=\Pr\left( L\geq \Theta \cap B_t>0\vert H_t=h  \right)\\
&=\Pr\left( L\geq \Theta\vert H_t=h \right)\Pr\left( B_t>0 \right)\\
&=q_{h}\left[1-\Pr\left( B_t=0 \right)\right]
\label{eq:Pu1}\end{split}\end{equation}
From \eqref{eq:BEH}, \eqref{eq:PY}, and \eqref{eq:Pu1} we observe that the BD of an energy harvesting sensor at time $t$ depends on the observation distribution at the sensor through $q_{h}$ and the battery depletion probability $\Pr\left( B_{t}=0 \right)$. Under the assumptions that the energy harvesting probability is i.i.d.\ Bernuoulli over time and space, and has a probability $p_e$, and the observations at the sensor are i.i.d.\ in time, we say that, the battery has a Markovian behavior in the sense that, conditioned on $E_t$ and $W_t$, its state at time $t+1$ (i.e., $B_{t+1}$) only depends on its state at time $t$ (i.e., $B_{t}$) and not the sequence of previous states, $\{B_{t'}\}_{t'=0}^{t-1}$. Under the Markovian assumption, a steady state probability for the battery charge state can be derived, which allows us to consider the long term performance of an energy harvesting sensor equipped with a battery. Let the state probability vector of the battery charge at time $t$ be 
\[{\bf{p}}_t\triangleq \big(\Pr\left(B_t=0\right),\ldots,\Pr\left(B_t=K\right) \big)^\intercal,\] 
where the superscript $^\intercal$ indicates transposition, and let us define the transition probability matrix as
\[{\bf P}=\big[p_{i,j} \big]_{\mathcal{K}\times \mathcal{K}}\,,\]
where $\mathcal{K}=K+1$ and $p_{i,j}\triangleq \Pr(B_{t}=j\vert B_{t-1}=i)\,,$
for $i,j=0,\ldots,K$. When the battery is in steady state \cite{Gro08}, the state probability vector will satisfy
\begin{equation}\label{eq:ssbb}
{\bf{p}}_{\infty}={\bf{P}^\intercal}\,{\bf{p}}_{\infty}\,,
\end{equation}
and the steady state probability of each battery state, i.e., \[p_{j}\triangleq \lim_{t\to \infty}\Pr(B_t=j),\] can be found using \eqref{eq:ssbb} and the fact that the summation of the probabilities must equal unity.

In steady state, after dropping the subscript $t$, the conditional probabilities in \eqref{eq:Pu1} are given by
\begin{equation}\begin{split}
P_{U\vert H}(u=1\vert h)&=q_h(1-p_{0})\,,\\
P_{U\vert H}(u=0\vert h)&=1-q_h(1-p_{0})\,.
\label{eq:Pu10}\end{split}\end{equation}
We can plug \eqref{eq:Pu10} into \eqref{eq:PY} and find the resulting Bhattacharyya distance in \eqref{eq:BEH} as
\begin{equation}\begin{split}
\mathcal{B}=-\log\Big[&\sqrt{[\epsilon_0+\delta q_0(1-p_0)][\epsilon_0+\delta q_1(1-p_0)]}\,+\\
&\sqrt{[1-\epsilon_0-\delta q_0(1-p_0)][1-\epsilon_0-\delta q_1(1-p_0)]} \,\Big].
\end{split}\label{eq:BDE}\end{equation}
We observe that the resulting BD for the energy constrained case (in steady state) depends on the depletion probability $p_0$ of the battery, which itself is a function of energy features and the battery capacity. Therefore, in the following we study the performance of an energy harvesting sensor for different battery capacities, $K$. To this end, for an arbitrary battery capacity $K$, in \textit{Lemma \ref{th:p0}}, we will find the depletion probability $p_0$. Then we can find a threshold $\Theta^\star$ which maximizes \eqref{eq:BDE}.

\begin{figure*}[!t]
\begin{minipage}[b]{0.45\linewidth}
\includegraphics[scale=0.77]{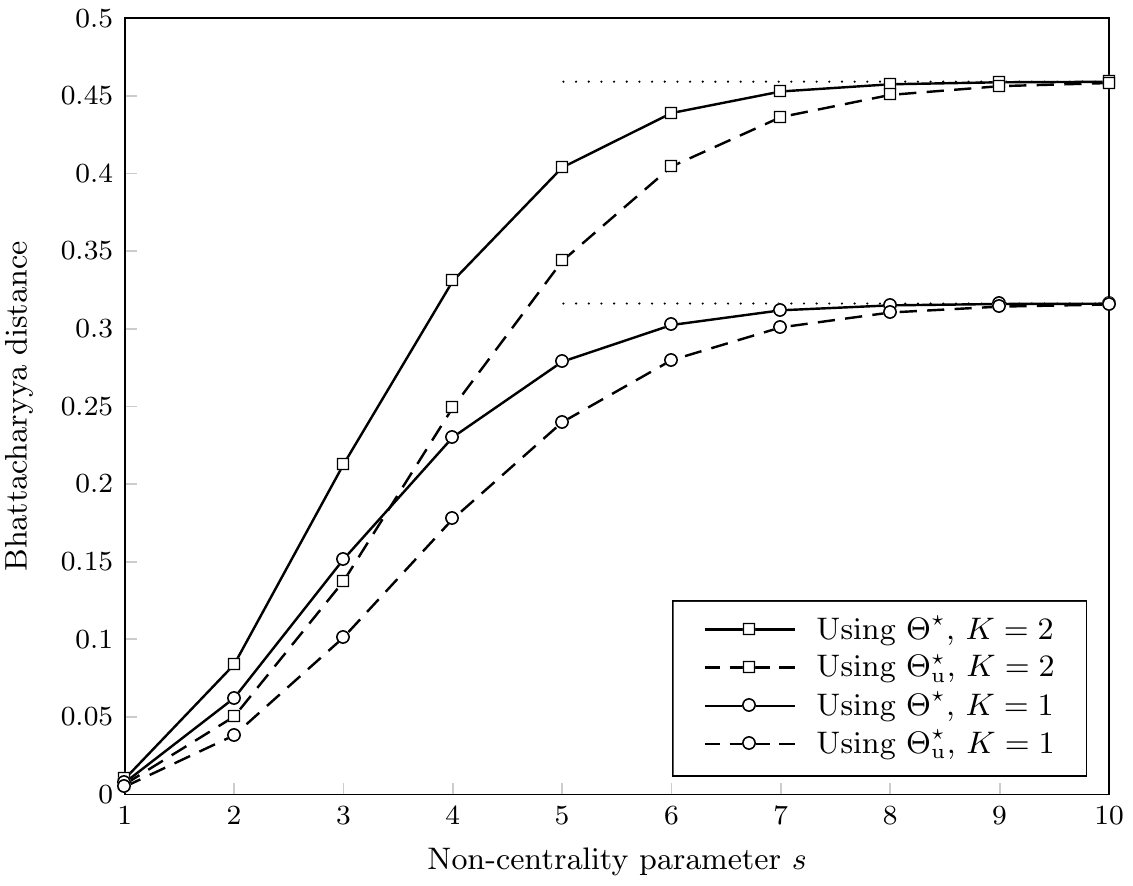}
\end{minipage}
\hspace{0.5cm}
\begin{minipage}[b]{0.45\linewidth}
\includegraphics[scale=0.77]{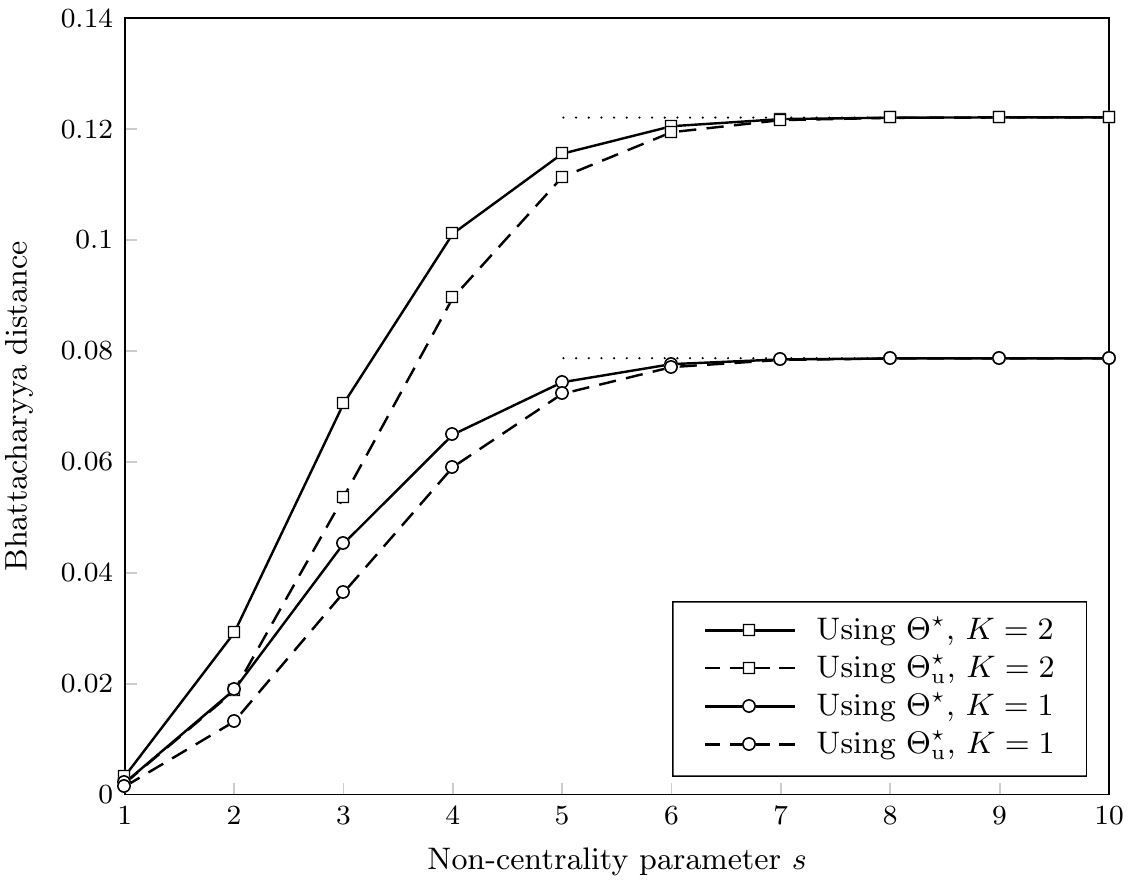}
\end{minipage}
\caption{Bhattacharyya distance of single- and double-slot-battery ($K=1,2$) energy harvesting sensors for different threshold tests and for noiseless (left) and noisy (right) communication channels, for $\pi_1=0.2$, $p_{e}=0.15$, as a function of non-centrality parameter.}
\label{fig:B12}
\end{figure*}

To this end, consider a $K$-slot-battery energy harvesting sensor $S$. Assume that at transmission time $t$, its battery charge is in state $b_t=k$, where $0< k< K$. Since during each time interval the sensor is capable of harvesting and consuming at most one packet of energy, its state at transmission time $t+1$ will be either $b_{t+1}=k-1, k$ or $k+1$. Two exceptions to this rule are when the battery charge is in state $b_t=0$ or in state $b_{t}=K$. In the former case the state of the battery at transmission time $t+1$ will be either zero or one, since the battery charge can not be negative. In the latter case, the state of the battery at transmission time $t+1$ will be either $K-1$ or $K$, since there is no space to save more energy packets. 
\begin{lemma}\label{th:p0}
The depletion probability of a $K$-slot-battery energy harvesting sensor is given by\footnote{We would like to note that the difference between this model and an M/M/1/K queue is that in this model each battery charge state $k$ after each transition can either change by one, or remain unchanged.}
\begin{subequations}\label{eq:pd}
\begin{align}
p_{0}&=\left[{1+\frac{1}{1-q}\sum_{k=1}^K \Omega^k}\right]^{-1}\\
&=\frac{p_{e}-q}{p_{e} \Omega^K-q}\,,
\end{align}\end{subequations}
where $$\Omega\triangleq \frac{p_{e}(1-q)}{q(1-p_{e})}\,,$$
and $q\triangleq \pi_0 q_0+\pi_1 q_1$.
\end{lemma}
\begin{proof}
See Appendix \ref{app:p0}.
\end{proof}

\textit{Remark:} The depletion probability of a $K$-slot-battery energy harvesting sensor (when $p_e<1$) is always above zero for any $K<\infty$. However, when $K=\infty$ it admits the following expression: 
\begin{equation}
p_{0}=\left \{
  \begin{tabular}{cc}
 $0$ & $p_{e}\geq q$\,,\\
 $1-\frac{p_e}{q}$ & Otherwise\,.
  \end{tabular}
\right.
\label{eq:pb0i}\end{equation}
This can be seen by noting that, the summation $$\sum_{k=1}^\infty\Omega^k=\sum_{k=1}^\infty \left(\frac{p_{e}(1-q)}{q(1-p_{e})}\right)^k$$ converges if $p_{e}< q$. Otherwise the sum diverges, which makes the depletion probability equal to zero. This depletion probability expression for an infinite capacity battery sensor is in line with the expression found before by modeling the battery state as a birth-death process in \cite{Alla16-EUSIPCO}.

Note that the condition $p_{e} \geq q$, which results in zero depletion probability, follows intuition in the sense that, when the probability of energy arrival $p_{e}$ is higher than the probability of energy consumption (or the probability $q$ that the sensor observation is above the threshold and it decides to send a message), the battery will accumulate energy and, with probability one, in long term not be empty. 
In this situation the problem will be the same as the unconstrained setup.

Now by plugging the expression for the depletion probability in \eqref{eq:pd} into that of the BD for energy constrained sensor in \eqref{eq:BDE}, we can find a closed-form expression for the BD of a $K$-slot-battery sensor as follows.
\begin{equation}\begin{split}
\mathcal{B}&=\!-\!\log\!\Bigg[\!\sqrt{\left[\epsilon_0+\delta q_0\frac{p_e(\Omega^K-1)}{p_e\Omega^K-q}\right]\!\!\left[\epsilon_0+\delta q_1\frac{p_e(\Omega^K-1)}{p_e\Omega^K-q}\right]}\,+\\
&\sqrt{\left[1-\epsilon_0-\delta q_0\frac{p_e(\Omega^K-1)}{p_e\Omega^K-q}\right]\!\!\left[1-\epsilon_0-\delta q_1\frac{p_e(\Omega^K-1)}{p_e\Omega^K-q}\right]} \,\Bigg].
\label{eq:BDE2}\end{split}\end{equation}
Note that since the BD is a function of observation models, we can not simplify it more. In what follows we will be considering different battery capacities $K$, and numerically find the threshold $\Theta^{\star}$ which maximizes the corresponding Bhattacharyya distance.

When a sensor is capable of saving only one packet of energy, incoming energy $e_t$ is saved in the battery if the battery is empty. Otherwise, the sensor discards the incoming energy packet. For a single-slot-battery sensor an optimal threshold $\Theta^\star$ is found by maximizing the BD in \eqref{eq:BDE2} for $K=1$. To do this, we use a grid search to find the threshold $\Theta$ that maximizes the BD. When a sensor is capable of saving two packets of energy ($K=2$), incoming energy $e_t$ is saved in the battery if the battery is empty or has only one packet of energy in the buffer. Otherwise, the sensor discards the incoming energy packet. The same as for a single-slot-battery, by maximizing \eqref{eq:BDE2} we can find an optimal threshold $\Theta^\star$ for a double-slot-battery sensor $K=2$. Fig.~\ref{fig:B12} illustrates the resulting BD of single- and double-slot-battery energy harvesting sensors using the adapted ($\Theta^\star$) and the unconstrained (${\Theta}^\star_{\mathrm{u}}$) thresholds, for $\pi_1=0.2$ and $p_e=0.15$. The figure on the left, considers error-free communication channels ($\epsilon_0=\epsilon_1=0$) while the figure on the right considers BAC channels with parameters $\epsilon_0=0.1$ and $\epsilon_1=0.2$.

We observe that the unconstrained threshold $\Theta_{\mathrm{u}}^{\star}$ is always sub-optimal. In the single- and double-slot-battery cases the Bhattacharyya distance is bounded when the SNR goes to infinity. In \textit{Theorem \ref{lm:BK}} we will find a closed-form expression for this upper bound for an arbitrary battery size $K$, and find conditions under which the Bhattacharyya distance is upper bounded.

\begin{theorem}\label{lm:BK}
Consider a $K$-slot-battery energy harvesting sensor $S$. Assume that the probability of harvesting energy at each time interval is $p_{e}$, and the a-priori probability of hypothesis $H_t=1$ is $\pi_1$ and the sensor-to-FC channels are BAC channels as in Fig.~\ref{fig:BAC}. The BD of this sensor at the input of the FC can not exceed 
\begin{equation}\label{eq:BK}
\begin{split}
\overline{\mathcal{B}}\triangleq-\log\Big[\sqrt{\epsilon_0\left(1-\epsilon_1-\overline{p}_0\delta \right)}\,+\sqrt{(1-\epsilon_0)\left(\epsilon_1+\overline{p}_0\delta \right)}\Big],
\end{split}\end{equation}
where
\[\overline{p}_0=\left[{1+\frac{1}{1-\pi_1}\sum_{k=1}^K \left(\frac{p_{e}(1-\pi_1)}{\pi_1(1-p_{e})}\right)^k}\right]^{-1}.\]
\end{theorem}
\begin{proof}
See Appendix \ref{app:BK}.
\end{proof}

\textit{Remark:} The upper bound in \eqref{eq:BK} is achievable under any separable observation model by selecting $\Theta^\star$ so that $q_0=0$ and $q_1=1$. It is also asymptotically achievable (when $\mathcal{E}\to\infty$) for a sensor with an asymptotically separable observation model. These asymptotes are also shown in Fig.~\ref{fig:B12} by dotted lines.
Note that a, possibly unexpected, insight due to this result is that when the sensors make (asymptotically) separable observations, i.e., when the sensors are sure of the true hypothesis, it is optimal for the sensor to act greedily and always transmit whenever $H_t=1$ and their internal battery allows for a transmission. The FC will in this situation still only receive a transmission from a subset of the sensors, as the battery depletion events are independent across sensors due to the independence of the energy arrival process.

\textit{Remark:} The upper bound in \eqref{eq:BK} is an increasing function of the battery capacity $K$. It is in line with the intuition that the performance of an energy harvesting sensor $S$ is improved by increasing its battery capacity. The upper bound in \eqref{eq:BK} is also an increasing function of the probability of harvesting energy $p_e$ and is a decreasing function of the a-priori probability $\pi_1$.

These also follow the intuition in the sense that by increasing the probability of having energy available at the battery, the performance of the sensor is improved. While $p_e$ affects the amount of available energy directly, $\pi_1$ affects the battery content in a more complicated way: According to \eqref{eq:trtest2} an optimally designed sensor aims to send a message ``1'' and consume a packet of energy when its observation is above $\Theta^\star$. By increasing the a-priori probability of the hypothesis $H_t=1$, it will be more likely that the sensor aims to send a ``1'' and consume energy. This itself increases the depletion probability $p_0$ and so decreases the performance of the sensor.

We would also like to note that, unlike previous works on the design of sensor decision functions in a distributed detection network using the Bhattacharyya distance as performance metric, our problem formulation comprises the affect of a-priori probabilities $\pi_j$, for $j=0,1$ through $q$ in the depletion probability $p_0$.

\textit{Remark:} For any finite $K$, the BD never grows unboundedly with $\mathcal{E}\to\infty$, while for the unconstrained case we have seen that it can grow unboundedly for separable and asymptotically separable observation models for error-free channel by increasing the SNR.

Using \eqref{eq:BDE2} and \eqref{eq:BK}, one can analyze the BD performance of a $K$-slot-battery sensor (when $K<\infty$) and its asymptote. In Fig.~\ref{fig:Bbmax} the optimum BD for an energy harvesting sensor as a function of the sensor battery capacity is shown when the observation model at the sensor is according to \eqref{eq:dist01} with $s=5$. We observe from this figure that the maximum BD and the upper bound (as discussed before) are both increasing (non-decreasing) functions of sensor battery capacity $K$.

\textit{Remark:} For an infinite battery capacity sensor $K=\infty$, under some conditions the BD distance grows unboundedly for an asymptotically separable model, as SNR increases. Concretely, only if $p_{e}\geq \pi_1$ and communication channels are noiseless, $\epsilon_0=\epsilon_1=0$, the Bhattacharyya distance of an optimal threshold test \eqref{eq:trtest2} increases unboundedly for separable or asymptotically separable observation models, as SNR increases.

In Fig.~\ref{fig:Bi} the Bhattacharyya distance for an infinite-slot-battery sensor is shown, when using different thresholds, and for different setups. We observe from this figure that as the non-centrality parameter (or SNR) increases, for the case where $p_e\geq \pi_1$ and $\epsilon_0=\epsilon_1=0$ the Bhattacharyya distance increases unboundedly, otherwise it is upper bounded with the asymptote (shown by the dotted line) found in \eqref{eq:BK}.

In the following section, we compare the error probability performance of networks of energy harvesting sensors which we have designed using conventional (unconstrained) formulation in \eqref{eq:BConven} (for $\Theta_{\mathrm{u}}^{\star}$) and using our proposed formulation in \eqref{eq:BDE2} (for $\Theta^{\star}$). Before considering the error probability performance of the designed networks, consider again the choice of $\Theta_{\mathrm{u}}^{\star}$ and $\Theta^{\star}$. While using $\Theta_{\mathrm{u}}^{\star}$, the sensor makes an optimal decision (in the sense of the Bhattacharyya distance) for the case where required energy for transmission of positive message is always available at the sensor. However, when energy is not always available at the sensor, the sensor should act more conservatively in the sense that: The sensor should remain silent and preserve energy for future time slots, unless it receives observation about the presence of hypothesis $1$ with high reliability. In our formulation for the Bhattacharyya distance, the threshold $\Theta^{\star}$ determines if a received observation has high reliability about the presence of hypothesis $1$. According to our simulation results, we always obtain $\Theta^{\star}\geq \Theta_{\mathrm{u}}^{\star}$ which further confirms our discussion above.

\begin{figure}[t]
\centering
\includegraphics[scale=0.8]{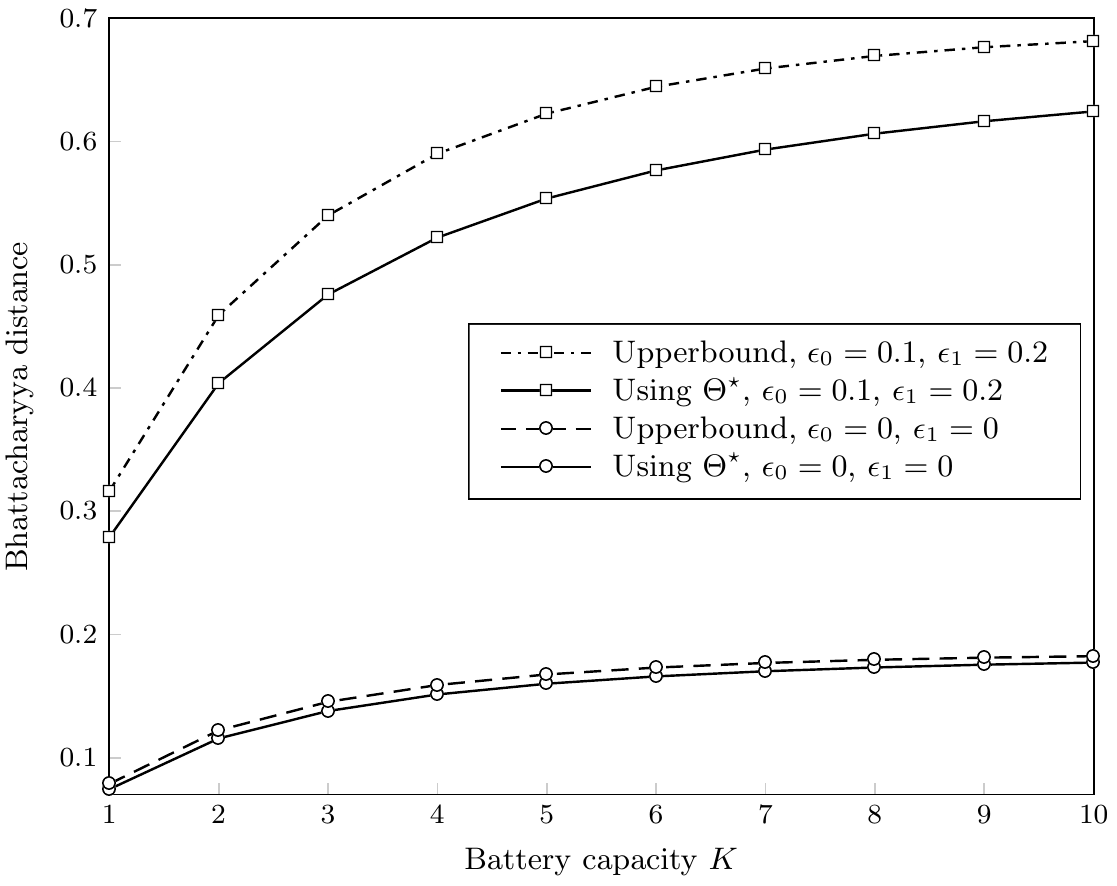}
\caption{Bhattacharyya distance of an energy harvesting sensor as a function of battery capacity $K$, for noisy and noiseless communication channels and for $\pi_1=0.2$, $p_e=0.15$, and corresponding upper bounds according to \textit{Theorem} \ref{lm:BK}.}
\label{fig:Bbmax}
\end{figure}
\section{Error Probability Performance of Networks}\label{sec:results}
In this section we illustrate the benefit of our results by numerical examples. Consider a sensor network with $N=4$ energy harvesting sensors. Suppose the sensors make observations from the same phenomenon $H_t$ during each time interval $t$, and send an OOK message to the FC. Let us assume again that sending a positive message consumes a packet of energy and a negative message is conveyed through a non-transmission with no energy cost. Let the observation model at each sensor be as in \eqref{eq:dist01}, and conditioned on the true hypothesis, the observations be independent, and that the sensors use the threshold test in \eqref{eq:trtest2}. Note that, though in this section we assume the same observation model at the sensors, our results and conclusions drawn through this work are generalized to \emph{non-identical} observation models at the sensors (and respectively non-identical sensor) case.

Using our results in the previous section, we design sensor decision rules $\Theta$ for different battery sizes, and compare their error probability performance with those of the unconstrained case. The expected error probability, when the FC uses the MAP criterion, at time $t$ is found using \cite{Alla14}
\begin{equation*}
P_{\mathrm{E},t}=1-\sum_{\underline{y}_t}\max_{j=0,1}\left\{ \pi_{j}P_{\underline{Y}\vert H}\left(\underline{y}_t\vert j\right) \right\},
\end{equation*}
which can be numerically computed, without the need for Monte-Carlo simulations.

\begin{figure}[t]
\centering
\includegraphics[scale=0.8]{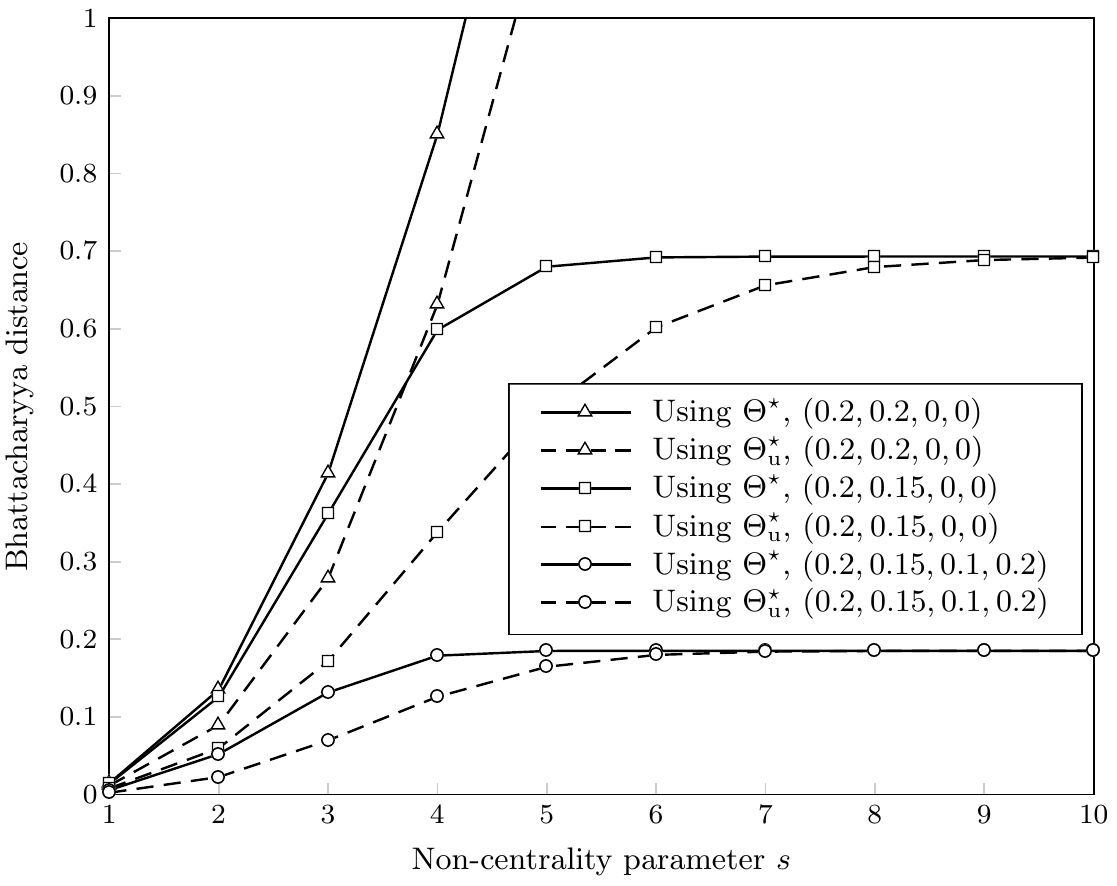}
\caption{Bhattacharyya distance of an infinite-slot-battery ($K=\infty$) energy harvesting sensor for different threshold tests and for different $(\pi_1,p_{e},\epsilon_0,\epsilon_1)$, as a function of noncentrality parameter.}
\label{fig:Bi}
\end{figure}

Fig.~\ref{fig:Pe12s} shows the error probability performance of designed sensor networks, with single-slot-battery sensors and double-slot-battery sensors, when $\pi_1=0.2$ and $p_{e}=0.15$. In both cases ($K=1,2$) the adapted threshold $\Theta^{\star}$ leads to better performance than the unconstrained threshold $\Theta^{\star}_{\mathrm{u}}$. This is in line with our results based on the Bhattacharyya distance: Adapted threshold leads to a higher BD. As $s\to\infty$, we also observe that the error probability does not converge to zero (lower bounded). It was shown for the Bhattacharyya distance that $\mathcal{B}_t$ is upper bounded for single and double-slot-battery sensors. We also observe from the figure that, while both using the optimal threshold or the sub-optimal threshold, increasing the battery capacity (here from one to two) can improve the error probability performance of a network of sensors.

In Fig.~\ref{fig:Pi1} the error probability performance of the network of $N=4$ infinite-slot-battery sensors is shown for different sets of $(\pi_1,p_e,\epsilon_0,\epsilon_1)$. When $\pi_1>p_e$ and/or sensor-to-FC channels are erroneous, the error probability of the network converges to a fixed value as $s\to\infty$, while when $\pi_1\leq p_e$ and the channels are error-free ($\epsilon_0=\epsilon_1=0$), the error probabilities for both the adapted threshold and the energy unconstrained threshold rapidly go to zero. These observations are also in line with our results in terms of the Bhattacharyya distance: When $K=\infty$, if $p_e\geq \pi_1$ and channels are error-free, the BD increases unboundedly, otherwise it converges to a non-zero asymptote. Note that we have the same observations for other choices of $(\pi_1,p_{e},\epsilon_0,\epsilon_1)$.

\begin{figure}[t]
\centering
\includegraphics[scale=0.8]{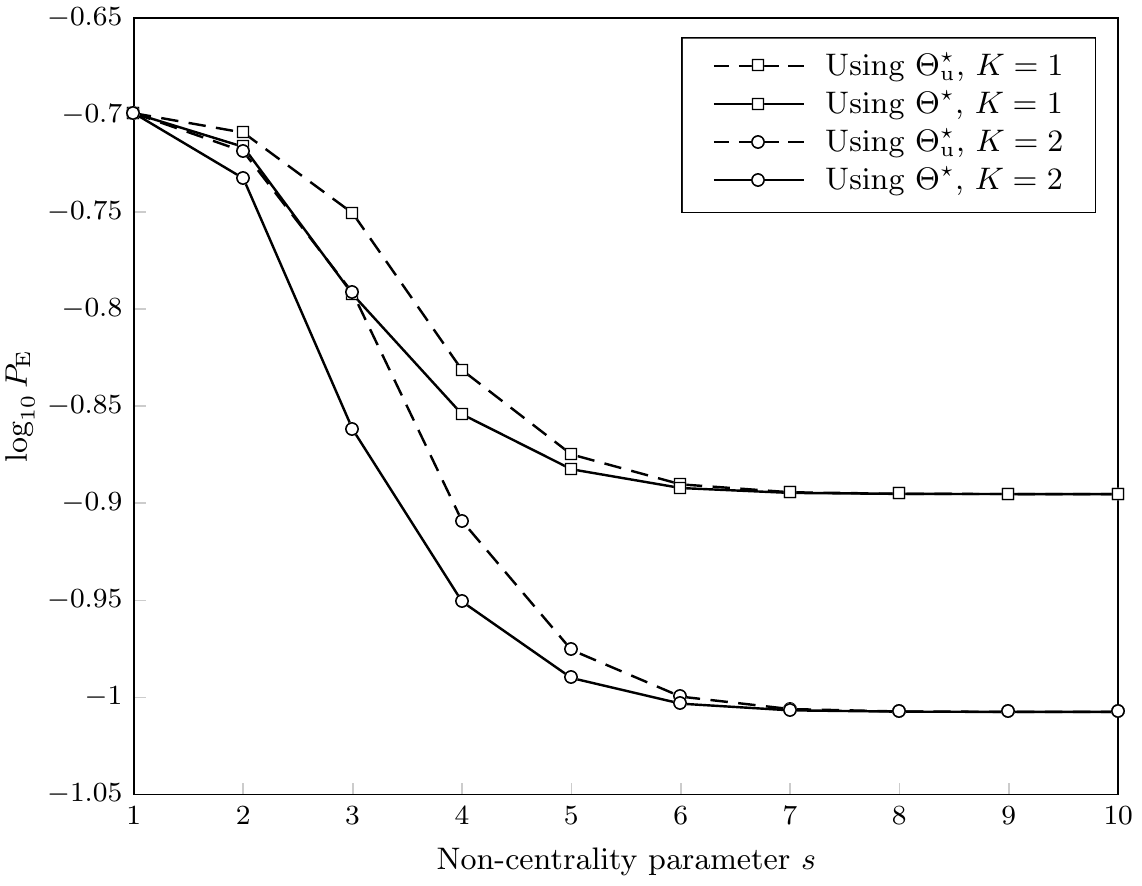}
\caption{Error probability performance of networks with $N=4$ energy harvesting sensors with different battery capacities, when $\pi_1=0.2$ and $p_{e}=0.15$, and noisy communication channels with $\epsilon_0=0.1$ and $\epsilon_1=0.2$.}
\label{fig:Pe12s}
\end{figure}

\section{Concluding Remarks}\label{sec:conclude}
In this paper we studied the problem of decentralized hypothesis testing in a network of energy harvesting sensors. The sensors in the network make observations of a phenomenon and harvest all the energy they need from the environment. We consider the case where the sensors have different battery capacities to save the harvested energy. Considering the Bhattacharyya distance as a performance metric, we formulated the problem of designing sensors in the network by considering the constraints imposed by energy harvesting and proposed a method to design the sensors decision rules. We further studied the performance of sensors for different battery capacities and presented conditions under which the Bhattacharyya distance is upper bounded, and therefore the error probability is lower bounded and does not converge to zero. 

In this paper, we considered the case where the observations and the energy charging processes at the sensors are independent. A possible extension to this work can be to consider the case where the observations at the sensors, or/and the energy charging processes are correlated. For the sake of analytical tractability and insight, we only studied the case where each sensor decision rule is a simple single-threshold test. It is known that having multiple-thresholds at the sensors can result in better BD performance. Concretely, battery state dependent thresholds $\Theta(b_t)$ can improve the performance of a single-threshold energy harvesting sensor, and an extension to our study can be to study the performance of such threshold tests. Although not presented herein, we have numerically found optimal threshold tests $\Theta_1^{\star},\Theta_2^{\star}$, where $\Theta_i^{\star}$ is the optimal threshold when $b_t=i$ for $i=1,2$, using a grid search in the case of a double-slot-battery ($K=2$) and the observation model in \eqref{eq:dist01}.
Our results confirm that having multiple thresholds can improve the BD performance of the sensor, but the gains are small compared to introducing the single threshold in the first place.

We would also like to note that other members of Ali-Silvey distance (like J-divergence) can also be used as a performance metric for the design of decision rules at the sensors, with a lot of similar derivations. In this paper, we chose to use the BD as the performance metric since according to \cite{Lon90} the BD is one of the most analytically tractable members of Ali-Silvey distances and was reported as the most efficient metric among those studied in \cite{Poor77}.

\begin{figure}[t]
\centering
\includegraphics[scale=0.8]{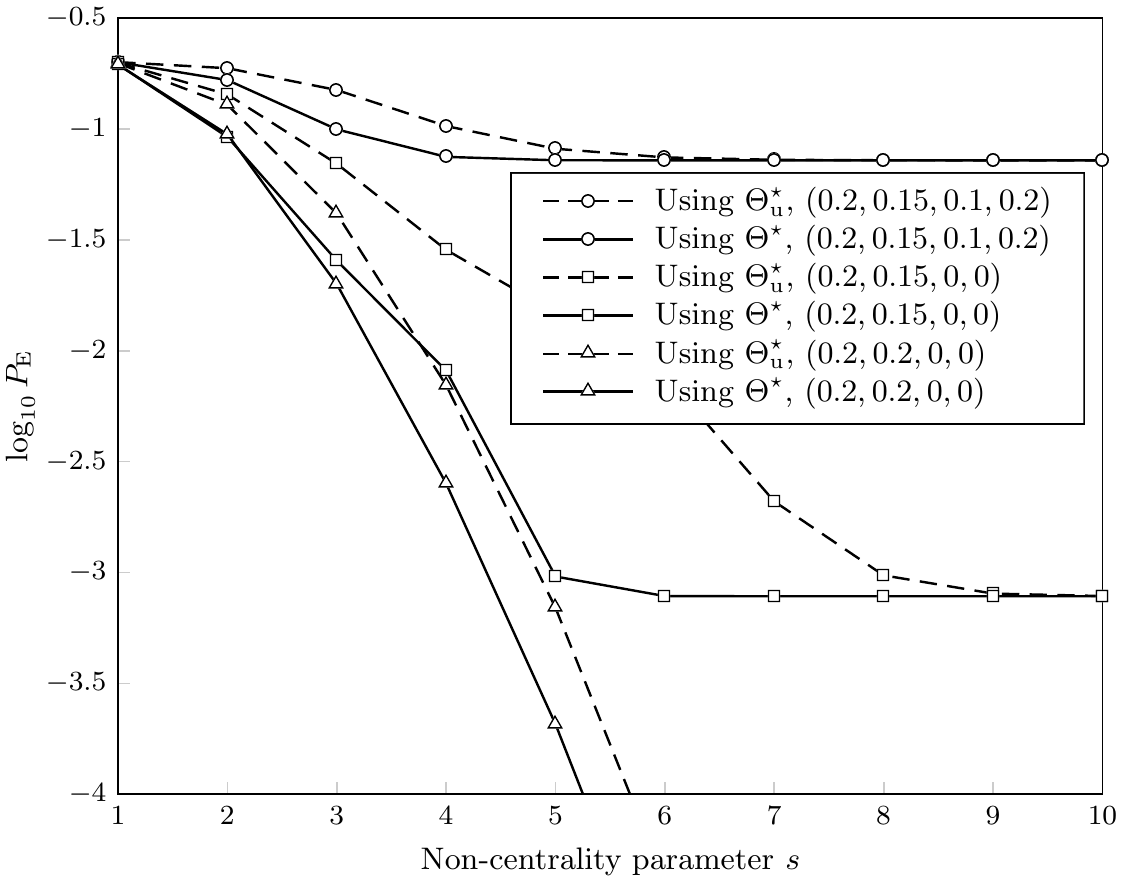}
\caption{Error probability performance of networks with infinite-slot-battery ($K=\infty$) energy harvesting sensors, and for different $(\pi_1,p_{e},\epsilon_0,\epsilon_1)$.}
\label{fig:Pi1}
\end{figure}

Our work represents a first attempt to introduce energy harvesting considerations in the context of distributed detection using established (but arguably simple) design metrics such as the Bhattacharyya distance. Our approach furthermore utilizes simplified models such as the energy arrival model taken from \cite{Valen16a,Miche15,Miche12a,Tutu14}. While the simplifying assumptions make the results more tractable and interpretable, they also naturally come with some strong limitations. Apart from the extension to a battery-dependent decision threshold, one could also consider the case where the sensors optimize the OOK transmission energy in order to influence the error probability of the transmission, i.e., the cross-over probability $\epsilon_0$ in Fig.~\ref{fig:BAC}, based on some given physical channel model. This would however arguably be more relevant in conjunction with a more refined energy arrival and storage model, which we cannot presently handle directly given that we (for simplicity) assume that energy arrives in quanta matched to the energy needed for a positive (on) transmission. Furthermore, although the decoupling of the sensor design caused by considering the Bhattacharyya distance is often used to simplify the sensor design problem in the distributed detection literature, this approach should be scrutinized again when adding the energy harvesting aspect. This is due to the fact that jointly designed sensor decision rules can potentially lead to tangible benefits through more advanced joint energy conservation rules across sensors than what is implicitly provided by the differences in the random battery state across the sensors.

\section*{Acknowledgment}
The authors would like to thank the Associate Editor and anonymous reviewers for their valuable comments and suggestions that led to improvement of this paper. The authors would especially like to thank the reviewers who prompted us to consider non-perfect communication channels.

\appendices
\section{Proof of Lemma \ref{th:p0}}\label{app:p0}
Using \eqref{eq:ssbb}, state probabilities for $p_0,\ldots,p_{K-1}$ are found as
\begin{equation}\begin{split}
p_{0}&= p_{0,0} p_{0}+p_{1,0} p_{1},\\
p_{{k}}&=p_{k-1,k} p_{{k-1}}+p_{k,k} p_{{k}}+p_{k+1,k} p_{{k+1}},\quad 1\leq k<K.
\label{eq:pkk}\end{split}\end{equation}
Using \eqref{eq:bt1} and \eqref{eq:trtest2}, for the transition probabilities we obtain
\[p_{0,0}=1-p_{e}\,,\quad p_{1,0}=q(1-p_{e})\,,\quad p_{0,1}=p_{e}\,,\]
and for $1\leq k< K$
\begin{equation}\begin{split}
p_{k,k+1}&=(1-q)p_{e}\,,\\
p_{k+1,k}&=q(1-p_{e})\,,\\
p_{k,k}&=qp_{e}+(1-q)(1-p_{e})\,.
\label{eq:pij}\end{split}\end{equation}
By re-arranging the equations in \eqref{eq:pkk} and replacing the transition probabilities $p_{i,j}$ with those in \eqref{eq:pij} we get the following equations 
\begin{subequations}\label{eq:diff}\begin{align}
p_{1}&=\frac{p_{e}}{q(1-p_{e})}\,p_{0} \label{eq:diffa}\\ 
p_{2}&=\frac{p_{e}^2(1-q)}{q^2(1-p_{e})^2}\,p_{0} \label{eq:diffb}\\
p_k&=-\Omega \,p_{k-2}+(1+\Omega)\,p_{k-1}\,,\quad 3 \leq k \leq K.
\end{align}\end{subequations}
The system of equations described by \eqref{eq:diff} is a homogeneous difference equation \cite{din10}. Its characteristic polynomial is 
$$\rho^2-(1+\Omega)\rho+\Omega=0\,,$$
whose roots are $\rho=\{1,\Omega\}$. These lead to the general solution for the difference equation as 
\begin{equation}p_k=\lambda_1+\lambda_2\Omega^k.\label{eq:charsol}\end{equation}
Now, applying auxiliary conditions \eqref{eq:diffa} and \eqref{eq:diffb}
to \eqref{eq:charsol}, we obtain $\lambda_1=0$ and $\lambda_2=\frac{1}{1-q}p_0$. Thus the solution to the difference equation is found to be given by 
\begin{equation}p_{k}=\frac{\Omega^k}{1-q}\,p_{0}, \quad 1\leq k \leq K\,,\label{eq:pk}\end{equation}
which describes each state probability $p_k$ in terms of the depletion probability $p_0$. Now, using the fact that the summation of the probabilities must equal unity, for a $K$-slot-battery we obtain the desired depletion probability expression \eqref{eq:pd}.

\section{Proof of Theorem \ref{lm:BK}}\label{app:BK}
To prove the theorem, first consider the following function $\Gamma_1$ where $0\leq \alpha\leq \beta\leq 1$, $0\leq \zeta\leq 1$, and $0\leq \xi\leq 0.5$. 
\begin{equation*}
\Gamma_1\triangleq\sqrt{[\xi+\zeta\alpha][\xi+\zeta\beta]}+\sqrt{[1-\xi-\zeta\alpha][1-\xi-\zeta\beta]}
\label{eq:proofBK2}\end{equation*}
Taking the derivative of $\Gamma_1$ with respect to $\alpha$, we observe that when $0\leq \alpha\leq\beta$, $\Gamma_1$ is an increasing function of $\alpha$ and its minima is attained when $\alpha=0$. In other words
\begin{equation*}
\min_{\substack{\alpha}}\Gamma_1=\sqrt{\xi[\xi+\zeta\beta]}+\sqrt{[1-\xi][1-\xi-\zeta\beta]}\,,
\label{eq:proofBK3}\end{equation*}
or equivalently
\begin{equation}\begin{split}
\sqrt{[\xi+\zeta\alpha][\xi+\zeta\beta]}+\sqrt{[1-\xi-\zeta\alpha][1-\xi-\zeta\beta]}\geq\\
\sqrt{\xi[\xi+\zeta\beta]}+\sqrt{[1-\xi][1-\xi-\zeta\beta]}\,
\label{eq:proofBK22}\end{split}\end{equation}
for $0\leq \alpha\leq \beta\leq 1$, $0\leq \zeta\leq 1$, and $0\leq \xi\leq 0.5$.

Using \eqref{eq:proofBK22}, we can conclude that for any given $q_1$, $\epsilon_0$, and $p_0$, the BD in \eqref{eq:BDE} is upperbounded by 
\begin{equation}\begin{split}
\mathcal{B}\leq-\log \Big[&\sqrt{\epsilon_0[\epsilon_0+q_1(1-p_0)\delta]}+\\
&\sqrt{[1-\epsilon_0][1-\epsilon_0-q_1(1-p_0)\delta]}\Big],
\label{eq:proofBK1}\end{split}\end{equation}
by replacing $\alpha$, $\beta$, $\xi$, and $\zeta$ in \eqref{eq:proofBK22} by $q_0$, $q_1$, $\epsilon_0$, and $(1-p_0)\delta$, respectively.

In what follows, we will find another upper bound for $\mathcal{B}$ by maximizing the right-hand-side of \eqref{eq:proofBK1} which is equivalent to minimizing $\Gamma_2$.
\begin{equation}\begin{split}
\Gamma_2\triangleq&\sqrt{\epsilon_0[\epsilon_0+q_1(1-p_0)\delta]}+\\
&\sqrt{[1-\epsilon_0][1-\epsilon_0-q_1(1-p_0)\delta]}
\end{split}\end{equation}
Minimizing $\Gamma_2$ is equivalent to maximizing $\Sigma\triangleq q_1(1-p_{0})$, since $\Gamma_2$ is a decreasing function of $\Sigma$. Then by replacing the depletion probability $p_0$ from \eqref{eq:pd} we obtain 
\begin{equation}\begin{split}
\Sigma=q_1\left(1-\left[{1+\sum_{k=1}^K \frac{p_{e}^k(1-q_0\pi_0-q_1\pi_1)^{k-1}}{(1-p_{e})^k(q_0\pi_0+q_1\pi_1)^k}}\right]^{-1}\right).
\label{eq:proofBK44}\end{split}\end{equation}

Now, considering the optimization problem
\begin{equation}
\max_{\substack{q_0,q_1\\0\leq q_0\leq q_1\leq 1 }}
\Sigma\,,
\label{eq:proofBK4}\end{equation}
the objective function $\Sigma$ is a decreasing function of $q_0$ (note that $q_0$ only appears in the argument of the summation) and its maxima is attained when $q_0=0$, and therefore the problem in \eqref{eq:proofBK4} reduces to the following problem,
\begin{equation}
\max_{\substack{0\leq q_1\leq 1 }}
\Lambda\,,
\label{eq:proofBK5}\end{equation}
where,
\begin{equation}\begin{split}
\Lambda&\triangleq \Sigma\,\Big\vert_{q_0=0}\\
&=q_1\left(1-\left[{1+\sum_{k=1}^K \frac{p_{e}^k(1-q_1\pi_1)^{k-1}}{(1-p_{e})^k(q_1\pi_1)^k}}\right]^{-1}\right).
\label{eq:proofBK6}\end{split}\end{equation}

In what follows we will show that the objective function $\Lambda$ is an increasing function of $q_1$ and therefore its maxima is attained when $q_1=1$, which proves the theorem. 
Our aim is now to show $$\frac{d\Lambda}{dq_1}\geq0,$$when $0\leq q_1\leq1$.  
By evaluating the summation in \eqref{eq:proofBK6} and introducing 
\begin{equation*}\begin{split}
\lambda\triangleq\frac{p_{e}(1-q_1\pi_1)}{(1-p_{e})q_1\pi_1}\,,
\label{eq:proofBK7}\end{split}\end{equation*}
we obtain 
\begin{equation*}\begin{split}
\Lambda=\left( \frac{p_e}{\pi_1} \right) \frac{1-\lambda^K}{1-p_e\lambda^K-(1-p_e)\lambda^{K+1}}\,.
\label{eq:proofBK8}\end{split}\end{equation*}
According to the chain rule we have $$\frac{d\Lambda}{dq_1}=\frac{d\Lambda}{d\lambda}\,\frac{d\lambda}{dq_1},$$ and using first principles we can straightforwardly show $$\frac{d\lambda}{dq_1}\leq 0,$$for $0\leq q_1\leq1$. To complete the proof, we need to show 
\begin{equation}
\frac{d\Lambda}{d\lambda}\leq 0\,,
\label{eq:proofBK9}\end{equation}
for $\lambda\geq0$. To this end, note that
$$
\frac{d\Lambda}{d\lambda} = \left( \frac{p_e}{\pi_1} \right) \frac{\lambda^{K-1}\big((K+1)\lambda-K-\lambda^{K+1}\big)(1-p)}{\big((1-p)\lambda^{K+1}+p\lambda^K-1\big)^2} \, .
$$
Thus, proving \eqref{eq:proofBK9} is equivalent to proving that
\begin{equation*}
g(\lambda)\triangleq (K+1)\lambda-K-\lambda^{K+1}\leq 0\,.
\label{eq:proofBK10}\end{equation*}
To prove this, we first show $g(\lambda)$ (for any $K\geq 0$) has its optima at $\lambda=1$ (where $g\left(1\right)=0$), by setting its first derivative $\frac{dg}{d\lambda}$ equal to zero, i.e., $$\frac{dg}{d\lambda}=(K+1)-(K+1)\lambda^K=(K+1)(1-\lambda^K)=0\,.$$
Then, we show this optima is a maxima (i.e., $\max_{\lambda}g(\lambda)=g(1)=0$) by finding its second derivative, i.e.,
$$\frac{d^2g}{d\lambda^2}=-K(K+1)\lambda^{K-1}\leq 0\,.$$
This completes the proof of $\frac{d\Lambda}{d\lambda}\leq 0$ and therefore the proof of Theorem \ref{lm:BK}.

\bibliographystyle{IEEEtran}
\bibliography{Ref}


\end{document}